\documentclass{article}
\usepackage{amssymb,amsmath,float
}

\usepackage[affil-it]{authblk}
\usepackage{color}
\usepackage[normalem]{ulem}

\usepackage[titletoc,title]{appendix}

\newcommand{\beq}{\begin{equation}}
\newcommand{\eeq}{\end{equation}}
\newcommand{\ba}{\begin{array}}
\newcommand{\ea}{\end{array}}
\newcommand{\bea}{\begin{eqnarray}}
\newcommand{\eea}{\end{eqnarray}}
\newcommand{\bean}{\begin{eqnarray*}}
\newcommand{\eean}{\end{eqnarray*}}

\newtheorem{theorem}{Theorem}

\newtheorem{example}[theorem]{Example}
\newtheorem{prop}[theorem]{Proposition}
\newtheorem{lem}[theorem]{Lemma}
\newtheorem{exe}[theorem]{Exercise}

\newtheorem{remark}[theorem]{Remark}
\newenvironment{rem}{\begin{remark} \rm}{\end{remark}}

\newtheorem{proof}{Proof.}
\newenvironment{pf}{\begin{proof} \rm}{\hfill$\square$ \end{proof}}

\newcommand{\mf}[1]{\mathfrak{#1}}
\newcommand{\bb}[1]{\mathbb{#1}}

\DeclareMathOperator{\Ht}{ht}
\DeclareMathOperator{\tr}{Tr}
\DeclareMathOperator{\ad}{ad}
\DeclareMathOperator{\diag}{diag}

\newcommand{\CW}{{\cal W}}

\newcommand{\CS}{{\cal S}}
\newcommand{\CN}{{\cal N}}

\newcommand{\CM}{{\cal M}}
\newcommand{\CQ}{{\cal Q}}

\begin{document}
\title{Poisson pencils: reduction, exactness, and invariants}
 \author{Paolo Lorenzoni$^1$,  Marco Pedroni$^2$, and Andrea Raimondo$^2$}
\affil{
{\small  $^1$Dipartimento di Matematica e Applicazioni\\
Universit\`a di Milano-Bicocca, Via Roberto Cozzi 55, I-20125 Milano, Italy}\\
{\small paolo.lorenzoni@unimib.it}\\
\medskip
{\small $^2$Dipartimento di Ingegneria Gestionale, dell'Informazione e della Produzione}\\
{\small Universit\`a di Bergamo, Viale Marconi 5, I-24044 Dalmine (BG), Italy}\\
{\small marco.pedroni@unibg.it, andrea.raimondo@unibg.it}
}
\date{}

\maketitle

\baselineskip=0,6cm

\begin{abstract}
\par\smallskip\noindent
We study the invariants (in particular, the central invariants) of suitable Poisson pencils from the point of view of the theory of bi-Hamiltonian reduction, paying a particular attention to the case where the Poisson pencil is exact. We show that the exactness is preserved by the reduction. 
In the Drinfeld-Sokolov case, the same is true for the characteristic polynomial of the pencil, which plays a crucial role in the definition of the central invariants. We also discuss the bi-Hamiltonian structures of a generalized Drinfeld-Sokolov hierarchy and of the Camassa-Holm equation.
\par\smallskip\noindent
{\it Keywords:\/} Drinfeld-Sokolov reduction; Poisson pencils of hydrodynamic type; central invariants; integrable PDEs; exact bi-Hamiltonian manifolds.
\end{abstract}

\section{Introduction}
Deformations of Poisson pencils $P_2-\lambda P_1$ of  hydrodynamic type
\beq\label{PPHT}
P^{ij}_a=g^{ij}_{a}({\bf w})\partial_x+\Gamma^{ij}_{k,a}({\bf w})w^k_x 
,\quad i,j=1,...,n,\quad\,a=1,2
\eeq
(see next section for more detailed definitions),
play a crucial role in the study of integrable hierarchies of PDEs \cite{DZ02,L2002,DLZ06}. In the semisimple case, this important class of Poisson pencils
 is parametrized by $n^2$ arbitrary functions of one variable. Part of these functions (namely, $n^2-n$) label semisimple Poisson pencils of hydrodynamic
  type \cite{F01} (see also \cite{M2002}). The remaining  $n$ functions, called \emph{central invariants},
   label Miura equivalent deformations of the same Poisson pencil of hydrodynamic type.   
   
   The most important examples are those for which the dispersionless limit is an exact Poisson pencil and the central invariants are constant \cite{DLZ17}.
    They include the Poisson pencils of the Drinfeld-Sokolov (DS) hierarchies \cite{DS85} corresponding to (the loop algebra of) a simple finite-dimensional Lie algebra 
$\mathfrak g$. In this case, it has been shown in \cite{DLZ08}
  that the dispersionless limit coincides with the Poisson pencil of hydrodynamic type associated with the flat pencil of metrics defined on the orbit space
   of the corresponding Weyl group (obtained using the Dubrovin-Saito construction \cite{D99,S93,SYS} for Coxeter groups), the central invariants are constant and their values coincide with the square of the lengths of the roots of $\mathfrak g$. For instance, in the simply laced case all the roots have equal lengths and, consequently, the associated integrable
 hierarchies are of topological type \cite{DZ02}.
  
  The aim of this paper is twofold. On one hand, we want to study the central invariants from the point of view of  the theory of bi-Hamiltonian reduction (that includes the DS  case). We will focus on the class of Poisson pencils which, after reduction, get the form of a deformation of a Poisson pencil of hydrodynamic type.
   In particular, in the DS case we will show that it is possible to define a set of invariants of the unreduced pencil that, once restricted on the  
   reduced manifold, define a set of functions which are equivalent to Dubrovin-Liu-Zhang central invariants of the reduced pencil. 
    
   On the other hand, we will focus on the case of exact Poisson pencils,  showing that the reduction process preserves the exactness of the pencil.
    Combining this result with the main result  of \cite{FL12}, relating the constancy of the central invariants with the exactness of the Poisson pencil,
     we have that the central invariants for the DS Poisson pencils are constant, in agreement with the Dubrovin-Liu-Zhang result.  

The paper is organized as follows. In Section 2 we recall various definitions concerning Poisson pencils of hydrodynamic type and their deformations. 
Section 3 is devoted to the bi-Hamiltonian reduction. In Section 4 we show that the exactness of a pencil is preserved by reduction. 
In Section 5 we collect some basic facts about the DS reduction, while in Section 6 we show that in this case the characteristic polynomial of the reduced pencil coincides with the one of the unreduced pencil. In the final section
 we consider two examples going beyond the DS framework.

\section{Invariants of Poisson pencils}\label{sec:central invariants}
We briefly summarize some basic notions about pencils of $N$-component infinite-dimensional Poisson brackets.
Consider \cite[Section $2$]{DZ02} (see also \cite{LZ04}) the following class of local Poisson brackets on the manifold of ${\mathbb R}^N$-valued functions ${\bf w}(x)=(w^1(x),\dots,w^N(x))$ over the unit circle $S^1$,
\begin{equation}\label{poisson}
\{w^{i}(x),w^{j}(y)\}=\sum_{k\ge -1}\epsilon^k\sum_{l=0}^{k+1}A^{ij}_{k,l}({\bf w},{\bf w}_x,\dots,{\bf w}_{(l)})\delta^{(k-l+1)}(x-y),
\end{equation}
where $A^{ij}_{k,l}$ are differential polynomials of degree $l$ (i.e., they are polynomials in the derivatives, whose coefficients are functions of ${\bf w}$), and $\deg f({\bf w})=0$,  $\deg ({\bf w}_{(l)})=l$. One can associate to \eqref{poisson} the differential operator
$$\Pi^{ij}=\sum_{k\ge -1}\epsilon^k\sum_{l=0}^{k+1}A^{ij}_{k,l}({\bf w},{\bf w}_x,\dots,{\bf w}_{(l)})\partial_x^{k-l+1},
$$
uniquely characterized by the relation $\{w^{i}(x),w^{j}(y)\}=\Pi^{ij}\delta(x-y)$. The most general group of transformations preserving the form of the bracket \eqref{poisson} is the group of Miura transformations.
 An element of this group is a transformation ${\bf w}\mapsto \tilde{\bf w}=(\tilde{w}^1,\dots,\tilde{w}^N)$, where
\begin{equation}\label{MT}
\tilde{w}^i= F_0^i({\bf w})+\sum_{k\ge1}\epsilon^k F^i_k({\bf w},{\bf w}_x,\dots,{\bf w}_{(k)}),\qquad{\rm det}\left(\frac{\partial F_0^i}{\partial w^j}\right)\ne 0,\, {\rm deg}F^i_k=k.
\end{equation}  
The Miura group can be described \cite{LZ06} as the semi-direct product of the subgroup of diffeomorphisms $Diff(\bb{R}^N)$ and the subgroup of Miura transformations starting from the identity:
\begin{equation}\label{MTid}
\tilde{w}^i= w^i+\sum_{k\ge1}\epsilon^k F^i_k({\bf w},{\bf w}_x,\dots,{\bf w}_{(k)}),\qquad\, {\rm deg}F^i_k=k.
\end{equation}  
The latter subgroup plays an important role in the classification theory of Poisson brackets: to prove that two brackets of type \eqref{poisson} are equivalent under a Miura transformation of type \eqref{MT} one first looks for a (local, in general) diffeomorphism mapping the leading term ($k=-1$) of the first bracket to the second, and then applies a transformation of type \eqref{MTid} --- which leaves the leading term invariant --- to obtain the second bracket.

After a transformation of type \eqref{MT}, the form of the Poisson bracket will be given by the operator
\begin{equation}\label{operatortrasformation}
\widetilde{\Pi}^{ij}=(L^\ast)^i_k\Pi^{ks}L^j_s,
\end{equation}
where 
$L$ and its adjoint $L^\ast$ are defined as
$$L^i_k=\sum_{s}(-\partial_x)^s\frac{\partial \tilde{w}^i}{\partial w^{k}_{(s)}},\qquad (L^\ast)^i_k=\sum_{s}\frac{\partial \tilde{w}^i}{\partial w^{k}_{(s)}}\partial_x^s.$$
A pair of Poisson brackets of the form \eqref{poisson},
\begin{equation}\label{bi-poisson}
\{w^{i}(x),w^{j}(x)\}_{a}=\sum_{k\ge -1}\epsilon^k\sum_{l=0}^{k+1}A^{ij}_{k,l;a}({\bf w},{\bf w}_x,\dots,{\bf w}_{(l)})\delta^{(k-l+1)}(x-y),\qquad a=1,2,
\end{equation}
is said to be \emph{compatible} if the pencil $\{\cdot,\cdot\}_{(\lambda)}=\{\cdot,\cdot\}_2-\lambda \{\cdot,\cdot\}_1$ is a Poisson bracket for any value of $\lambda$. In this case the bracket $\{\cdot,\cdot\}_{(\lambda)}$ is said to be a Poisson pencil. Given a pair of Poisson brackets \eqref{bi-poisson} one can consider the action of the Miura group \eqref{MT} on the pencil and define a set of invariants by means of the following recipe \cite{DLZ08,DVLS}. Introducing the differential operators
\begin{equation}\label{bipoisson-operator}
\Pi^{ij}_a
=\sum_{k\ge -1}\epsilon^k\sum_{l=0}^{k+1}A^{ij}_{k,l,a}({\bf w},{\bf w}_x,\dots,{\bf w}_{(l)})\partial_x^{k-l+1},\qquad a=1,2,
\end{equation}
then one defines the following power series\footnote{Note that our definition of \eqref{poisson-symbol} differs from the one given in \cite{DLZ08} by a  multiplicative factor $p$.} in the parameter $p$:
\begin{equation}\label{poisson-symbol}
\pi^{ij}_a(p;{\bf w})=\sum_{k\ge -1}A^{ij}_{k,0,a}({\bf w})p^{k+1},\qquad a=1,2.
\end{equation}
Note that \eqref{poisson-symbol} is strictly related (but it is not equal) to the symbol of \eqref{bipoisson-operator}, and that the coefficients of \eqref{poisson-symbol} do not depend on the derivatives of ${\bf w}$.   The 
rule \eqref{operatortrasformation} induces the following transformation on the pencil $\pi_{\lambda}^{ij}=\pi^{ij}_{2}-\lambda\pi^{ij}_{1}$ under Miura transformations \eqref{MT}:
  \begin{equation}\label{TRPP}
  \tilde\pi^{ij}_{\lambda}=l^i_h(p)\pi_{\lambda}^{hk}l^j_k(-p),
  \end{equation}
where 
$$l^i_j(p)=\sum_{k}\frac{\partial F^i_k}{\partial w^j_{(k)}}p^k.$$  
Since by definition \eqref{MT} we have that $\deg F^i_k=k$, then the quantities $l^i_j(p)$ may depend on ${\bf w}$ but not on the derivatives. 
Consider now the \emph{characteristic polynomial} of the Poisson pencil,
\begin{equation}\label{charpoly}
{\mathcal R}(p,\lambda;{\bf w})=\det\left(\pi^{ij}_{2}-\lambda\pi^{ij}_{1}\right)=
\det\left(\sum_{k\ge -1}\left(A^{ij}_{k,0;2}({\bf w})-\lambda A^{ij}_{k,0;1}({\bf w})\right)p^{k+1}\right),
\end{equation}
which is a polynomial, of degree say $M$, with respect to $\lambda$. In general, $M$ does not need to be equal to the dimension $N$. Due to the above construction, the functions $\lambda^i({\bf w},p)$, $i=1,\dots,M$, which are defined to be the $\lambda-$roots of the equation ${\mathcal R}(p,\lambda;{\bf w})=0$, are invariant under Miura transformations of type \eqref{MTid}, while they behave as scalars ($0-$tensors) with respect to the whole Miura group \eqref{MT}. 

\begin{example}
\label{exa:kdv-sec2}
Consider the Poisson pencil 
\begin{equation}
\label{sl2-sec2}
\begin{pmatrix} 0 & -2 \epsilon^{-1}(w^1-\lambda) & \epsilon^{-1}w^2-\partial_x\\ 2\epsilon^{-1}(w^1-\lambda) & -2\partial_x & -2\epsilon^{-1}w^3 
\\ -\epsilon^{-1}w^2-\partial_x & 2\epsilon^{-1}w^3 & 0\end{pmatrix}
\end{equation}
in the variables $(w^1,w^2,w^3)$. (We will see in Section \ref{sec:DS-structures} that it is the particular case of (\ref{DS-pb-delta}) corresponding to 
${\mathfrak{sl}}(2)$). The power series associated to this pencil is
\begin{equation}
\label{sl2symbol-sec2}
\begin{pmatrix} 0 & -2(w^1-\lambda) & w^2-p\\ 2(w^1-\lambda) & -2p & -2w^3 \\ -w^2-p & 2w^3 & 0\end{pmatrix},
\end{equation}
so that the characteristic polynomial (\ref{charpoly}), restricted at the points with $w^3=1$, is 
\begin{equation}
\label{charpoly-sl2-sec2}
{\mathcal R}
(p,\lambda;w^1,w^2)=2p^3-8\left(w^1+\frac14\left(w^2\right)^2-\lambda\right)p. 
\end{equation}
We will see in Example \ref{exa:kdv-ci} that it projects (up to a multiplicative constant) to the characteristic polynomial (\ref{kdv-charpoly}) of the KdV Poisson pencil.
\end{example}

In the remaining part of this section we will recall some important facts about Poisson pencils of the form \eqref{poisson} admitting semisimple dispersionless limit. Some preliminary results of \cite{DVLS} suggest that part of this theory might be generalized in the non-semisimple setting.

\paragraph{Dispersionless limit.} We now consider a special class of Poisson brackets \eqref{poisson}, namely 
the class
\begin{equation}\label{poissondisp}
\{w^{i}(x),w^{j}(y)\}=\sum_{k\ge 0}\epsilon^k\sum_{l=0}^{k+1}A^{ij}_{k,l}({\bf w},{\bf w}_x,\dots,{\bf w}_{(l)})\delta^{(k-l+1)}(x-y),
\end{equation}
where for consistency with the rest of the paper the number of components is now denoted by $n$ rather than $N$, so that ${\bf w}=(w^1,\dots,w^n)$. The class of brackets \eqref{poissondisp} admits the limit as $\epsilon\to 0$, known as \emph{dispersionless limit}, and  the leading term of the bracket,
\begin{equation}\label{pbht}
A^{ij}_{0,0}({\bf w})\delta'(x-y)+A^{ij}_{0,1}({\bf w},{\bf w}_x)\delta(x-y)=g^{ij}({\bf w})\delta'(x-y)+\Gamma^{ij}_k({\bf w})w^k_x\delta(x-y),
\end{equation}
is called {\it  Poisson bracket of hydrodynamic type}. Thus the dispersionless limit of a bracket of type \eqref{poissondisp} is a bracket of hydrodynamic type \eqref{pbht}. In the case where the matrix $g^{ij}$ is invertible, it defines a contravariant (pseudo)metric on 
${\mathbb R}^n$, which is flat and has $\Gamma^{ij}_k
$ as contravariant Christoffel symbols \cite{DN84}.

For Poisson brackets of the form \eqref{poissondisp} an analogue of the classical Darboux theorem holds true, as these brackets can be reduced to the constant form $\eta^{ij}\delta'(x-y)$
by means of Miura transformations \eqref{MT}. For Poisson brackets of hydrodynamic type, the Darboux coordinates are flat coordinates of the metric $g$. In the case of general Poisson brackets 
 \eqref{poissondisp}, in order to reduce the bracket to constant form one needs to reduce the bracket to its dispersionless limit. The existence of this reducing transformation --- which is of type \eqref{MTid} ---  was proved in \cite{G,DMS,DZ02}. 

In analogy to the general case, we also  consider compatible Poisson brackets of the form
\begin{equation}\label{bi-poissondisp}
\{w^{i}(x),w^{j}(y)\}_a=\sum_{k\ge 0}\epsilon^k\sum_{l=0}^{k+1}A^{ij}_{k,l,a}({\bf w},{\bf w}_x,\dots,{\bf w}_{(l)})\delta^{(k-l+1)}(x-y),\qquad a=1,2,
\end{equation}
with
\begin{equation}\label{bi-pbht}
A^{ij}_{0,0,a}({\bf w})\delta'(x-y)+A^{ij}_{0,1,a}({\bf w},{\bf w}_x)\delta(x-y)=g^{ij}_a({\bf w})\delta'(x-y)+\Gamma^{ij}_{k,a}({\bf w}){w^k}_x\delta(x-y),
\end{equation}
and the corresponding pencil $\{\cdot,\cdot\}_{(\lambda)}=\{\cdot,\cdot\}_2-\lambda \{\cdot,\cdot\}_1$. The dispersionless limit of this pencil is known as \emph{Poisson pencil of hydrodynamic type}.  The compatibility of the Poisson brackets implies that the pencil of contravariant metrics 
\begin{equation}\label{pencilmetric}
g^{ij}_{(\lambda)}=g^{ij}_{2}-\lambda g^{ij}_{1}
\end{equation}
is flat for any $\lambda$ and that the contravariant Christoffel symbols of $g^{ij}_{(\lambda)}$ are given by the pencil $\Gamma^{ij}_{k;2}-\lambda \Gamma^{ij}_{k;1}$
of the contravariant Christoffel symbols.  A pencil of contravariant flat metrics satisfying these conditions is called a \emph{flat pencil of metrics} \cite{D98}.  In general, it is not possible to reduce a Poisson pencil \eqref{bi-poissondisp} to its dispersionless limit by means of Miura transformations. If this happens, the pencil is said to be \emph{trivial}. 

\begin{example}
\label{exa:scalarcase}
Let us consider the second order deformations of the scalar Poisson pencil of hydrodynamic type 
$$2(u-\lambda)\delta'(x-y)+u_x\delta(x-y).$$ 
Using Miura transformations they can be reduced to the following form, 
\begin{equation}\label{scalar} 
2(u-\lambda)\delta'(x-y)+u_x\delta(x-y)+\epsilon^2\left(2c(u)\delta'''(x-y)+3c_x\delta''(x-y)+c_{xx}\delta'(x-y)\right)+{O}(\epsilon^4),
\end{equation}
where $c(u)$ is an arbitrary function. It turns out that any non vanishing function $c(u)$ defines a non trivial deformation \cite{L2002}.
\end{example}

\paragraph{Semisimplicity and central invariants.}
A flat pencil of metrics \eqref{pencilmetric} and the corresponding Poisson pencil of hydrodynamic type are said to be \emph{semisimple}  
if the $\lambda$-roots $u^1({\bf w}),\dots,u^n({\bf w})$ of 
$$\det\left(g^{ij}_{2}({\bf w})-\lambda g^{ij}_{1}({\bf w})\right)$$ 
are pairwise distinct  and nonconstant.  In this case, they 
form a set of coordinates ${\bf u}=(u^1,\dots,u^n)$, known as {\it canonical coordinates}, in which the two metrics are diagonal:
\begin{equation}\label{diagometrics}
g^{ij}_{1}({\bf u})=f^i({\bf u})\delta_{ij},\qquad g^{ij}_{2}({\bf u})=u^i f^i({\bf u})\delta_{ij},
\end{equation}
for some functions $f^i({\bf u})$. Note that by construction the functions $f^i({\bf u})$ are invariant under Miura transformations of type \eqref{MTid}. In the semisimple $n$-component  case the triviality of the pencil (\ref{bi-poissondisp}) is controlled by $n$ functions of a single variable, called \emph{central invariants}  and defined in the following way \cite{DLZ08}. Due to \eqref{poissondisp} and \eqref{pbht}, the determinant ${\mathcal R}(p,\lambda;{\bf w})$ defined in \eqref{charpoly} reads
\begin{align*}
{\mathcal R}(p,\lambda;{\bf w})&=\det\left(\pi^{ij}_{2}-\lambda\pi^{ij}_{1}\right)=
\det\left(\sum_{k\ge 0}\left(A^{ij}_{k,0;2}({\bf w})-\lambda A^{ij}_{k,0;1}({\bf w})\right)p^{k+1}\right)\\
&=p^n\det\left(g^{ij}_{2}({\bf w})-\lambda g^{ij}_{1}({\bf w})\right)+O(p^{n+1}).
\end{align*}
As in the general case, the roots $\lambda=\lambda^i(p;{\bf w})$ of the equation ${\mathcal R}(p,\lambda;{\bf w})=0$ are invariant with respect to Miura transformations of type \eqref{MTid}. Under the semisimplicity assumption, it can be shown \cite{DLZ08}  that the formal power series obtained expanding $\lambda^i(p;{\bf w})$ at $p=0$ contain only even powers of $p$ (this is not true anymore in the non-semisimple case \cite{DVLS}), and that the leading terms $u^i$ of the above series are the canonical coordinates:
\begin{equation}
\label{lambdaexpa}
\lambda^i(p;{\bf w})=u^i({\bf w})+\lambda^i_2({\bf w})p^2+{O}(p^4),\qquad i=1,...,n.
\end{equation}
All coefficients appearing in the series above are invariant under Miura maps of type \eqref{MTid}. Recall that the functions $f^i$, $i=1,\dots,n$, appearing in \eqref{diagometrics} are also unaffected by Miura transformations. The central invariants of the Poisson pencil are defined as \cite{DLZ06,DLZ08}
\begin{equation}
\label{central-inv}
c_i=\left.\frac{\lambda^i_2({\bf w})}{3f^i({\bf w})}\right|_{{\bf w}={\bf w}({\bf u})},
\end{equation}
and are thus an equivalent set of invariants. Once written in terms of the canonical coordinates, the function $c_i$ turns out to depend only on the coordinate $u^i$. Trivial Poisson pencils
 are characterized by the vanishing of all central invariants. More generally, it turns out that in the semisimple case two Poisson pencils 
  are related by a Miura transformation if and only if they have the same dispersionless part (in canonical coordinates) and the same central invariants \cite{DLZ06} (see also \cite{CKS} for an alternative proof). 
  
\begin{example}
\label{exa:kdv-ci}
Let us consider the well known Poisson pencil of the KdV hierarchy:
\begin{equation}
\label{kdv-pb-red}
\{u(x),u(y)\}_1=-2\delta'(x-y),\qquad \{u(x),u(y)\}_2=-u_x\delta(x-y)-2u\delta'(x-y)+\frac12\epsilon^2\delta'''(x-y).
\end{equation}
We have that 
\begin{equation}
\label{kdv-charpoly}
{\mathcal R}(p,\lambda;u)=\pi^{11}_{2}-\lambda\pi^{11}_{1}=-2up+\frac12 p^3+2\lambda p,
\end{equation}
so that (\ref{lambdaexpa}) takes the form 
$\lambda=u-\frac14 p^2$. Since $f^1(u)=-2$ and $\lambda^i_2(u)=-\frac14$, the central invariant is given 
by $c_1(u)=\frac1{24}$. In the scalar case (see Example \ref{exa:scalarcase}), up to a constant factor, the central invariant coincides with the function $c(u)$ 
appearing in 
formula \eqref{scalar}.
\end{example}

Notice that (\ref{kdv-charpoly}) is $1/4$ of the characteristic polynomial (\ref{charpoly-sl2-sec2}) obtained in Example \ref{exa:kdv-sec2}, if 
$u=w^1+\frac14\left(w^2\right)^2$. As we will show in Section \ref{sec:centralinvariants} (in the general context of the Drinfeld-Sokolov reduction), the reason is that the Poisson pencil in Example \ref{exa:kdv-ci} is the reduction of the one in Example \ref{exa:kdv-sec2}. 

\section{Some facts about bi-Hamiltonian reduction}
\label{sec:bi-Ham-reduction}

In this section we recall a general reduction theorem for bi-Hamiltonian manifolds (see \cite{CP92,pondi,MarleNunes} for details and proofs), and we prove a crucial result in order to show  the equality between the characteristic polynomials (\ref{charpoly}) of the reduced and unreduced bi-Hamiltonian structures of the form (\ref{bi-poisson}). 

Let $(\CM,\{\cdot,\cdot\}_1,\{\cdot,\cdot\}_2)$ 
be a bi-Hamiltonian manifold. The first step is to consider the (generalized) integrable distribution $D=P_2(\mbox{Ker}\,P_1)$, where $P_a$ is the Poisson tensor associated with $\{\cdot,\cdot\}_a$ by means of $\{F,G\}_a=\langle dG,P_a dF\rangle$. Then we choose a symplectic leaf $\CS$ of $P_1$ and we introduce the distribution on $\CS$ given by $E=D\cap T\CS$. If the quotient space
$\CN=\CS/E$ (whose points are the integral leaves of the distribution $E$) is regular, then it is a bi-Hamiltonian manifold. We call $(P'_1,P'_2)$ the reduced Poisson pair, and $(\{\cdot,\cdot\}'_1,\{\cdot,\cdot\}'_2)$ the corresponding Poisson brackets. They are given by
\begin{equation}
\label{red-pb}
\{f,g\}'_a(\pi({\bf w}))=\{F,G\}_a({\bf w}),\qquad a=1,2,\qquad {\bf w}\in\CS,
\end{equation}
where $\pi:\CS\to\CN$ is the projection on the quotient manifold and $F,G$ are extensions of $f\circ\pi,g\circ\pi$ on $\CM$ such that $(dF)_{\bf w},(dG)_{\bf w}$ vanish on the tangent vectors in $D_{\bf w}$ for all ${\bf w}\in\CS$. 

If an explicit description of the quotient manifold is not available, the following technique can be very useful. Suppose $\CQ$ to be a submanifold of $\CS$ which is transversal to the distribution $E$, in the sense that
\begin{equation}
\label{split}
T_{\bf w}\CQ\oplus E_{\bf w}=T_{\bf w}\CS\qquad\mbox{for all ${\bf w}\in\CQ$} .
\end{equation}
Then $\CQ$ also inherits a bi-Hamiltonian structure from $\CM$. The reduced Poisson brackets on $\CQ$ are given by
\begin{equation}
\label{red-pb-Q}
\{f,g\}'_a({\bf w})=\{F,G\}_a({\bf w}),\qquad a=1,2,\qquad {\bf w}\in\CQ,
\end{equation}
where $F,G$ are extensions of $f,g$ on $\CM$ such that $(dF)_{\bf w},(dG)_{\bf w}$ vanish on the tangent vectors in $D_{\bf w}$ for all ${\bf w}\in\CQ$. If the quotient $\CN$ is a manifold, then there is a local diffeomorphism from $\CQ$ to (an open subset of) $\CN$, connecting the Poisson pairs of the two manifolds. Notice however that we can define a reduced Poisson pair directly on $\CQ$, even in the case where the quotient $\CN$ is not a manifold. 

In terms of the pencil of Poisson tensors $P_{(\lambda)}=P_2-\lambda P_1$, the construction of the reduced Poisson structure on $\CQ$ goes as follows. 
Given ${\bf w}\in\CQ$ and ${\bf v}\in T_{\bf w}^*\CQ$, we look for an extension $\widehat{\bf v}\in T_{\bf w}^*\CM$ of ${\bf v}$ such that 
$\left(P_{(\lambda)}\right)_{\bf w}\widehat{\bf v}\in T_{\bf w}\CQ$. The existence of such a $\widehat{\bf v}$ has been shown in \cite{CP92}. The proof of its uniqueness can be found in \cite{CFMP98}, under the assumption that $\ker P_1\cap\ker P_2$ is trivial at the points of $\CQ$, which is true if there is no common Casimir (this situation is sometimes referred to as the non-resonant case). Independently of the uniqueness of $\widehat{\bf v}$, the reduced Poisson pencil is given by $\left(P'_{(\lambda)}\right)_{\bf w}{\bf v}=\left(P_{(\lambda)}\right)_{\bf w}\widehat{\bf v}$.

Now let us introduce coordinates on $\CM$ adapted to the transversal 
submanifold $\CQ$. More precisely, the first part of the coordinates can be seen as coordinates on $\CQ$, which is found by setting to zero the second part of them. Correspondingly, the matrix representing $P_{(\lambda)}$ can be decomposed as
$$
\begin{pmatrix}
A_{(\lambda)} & B_{(\lambda)}\\ C_{(\lambda)} & D_{(\lambda)}
\end{pmatrix}.
$$
For simplicity, we will sometimes use 
the same notations for geometric objects and matrices/vectors representing them in the chosen coordinate system. 
\begin{prop} 
\label{prop:schur} 
Suppose that $(\ker P_1)_{\bf w}\cap (\ker P_2)_{\bf w}=\{0\}$ for all ${\bf w}\in\CQ$. Then  
\begin{enumerate}
\item[a)] The matrix $D_{(\lambda)}$ is invertible.
\item[b)] The matrix representing the reduced Poisson pencil $P'_{(\lambda)}$ is given by $A_{(\lambda)}-B_{(\lambda)}D_{(\lambda)}^{-1}C_{(\lambda)}$. 
\item[c)] The identity
\begin{equation}
\label{decomp-poisson}
P_{(\lambda)}=\begin{pmatrix}  P'_{(\lambda)} & B_{(\lambda)}\\ 0 & D_{(\lambda)}\end{pmatrix}
\begin{pmatrix} \mbox{Id} & 0\\ D_{(\lambda)}^{-1}C_{(\lambda)} & \mbox{Id}
\end{pmatrix}
\end{equation}
holds true, where $\mbox{Id}$ is the identity matrix of the appropriate order.
\end{enumerate}
\end{prop}

\begin{pf}
a) As we wrote above, it was shown in \cite{CFMP98} that there is a unique extension $\widehat{\bf v}\in T_{\bf w}^*\CM$ of a given 
${\bf v}\in T_{\bf w}^*\CQ$ such that $\left(P_{(\lambda)}\right)_{\bf w}\widehat{\bf v}\in T_{\bf w}\CQ$, under the assumption that 
$(\ker P_1)_{\bf w}\cap (\ker P_2)_{\bf w}=\{0\}$. In \cite{Dinar} it is shown that the uniqueness of the extension is equivalent to the invertibility of $D_{(\lambda)}$. For the ease of the reader, we report here the proof.\\
Given ${\bf v}\in T_{\bf w}^*\CQ$, let $v$ be its components and $\left(\begin{smallmatrix} v\\ V\end{smallmatrix}\right)$ 
be the components of an extension $\widehat{\bf v}\in T_{\bf w}^*\CM$ such that $\left(P_{(\lambda)}\right)_{\bf w}\widehat{\bf v}\in T_{\bf w}\CQ$. Then
$$
\begin{pmatrix} P'_{(\lambda)}{v}\\ 0\end{pmatrix}=P_{(\lambda)}\begin{pmatrix} v\\ V\end{pmatrix}=
\begin{pmatrix} A_{(\lambda)} & B_{(\lambda)}\\ C_{(\lambda)} & D_{(\lambda)}\end{pmatrix}
\begin{pmatrix} v\\ V\end{pmatrix}=
\begin{pmatrix} A_{(\lambda)}v+B_{(\lambda)}V\\ C_{(\lambda)}v+D_{(\lambda)}V\end{pmatrix},
$$
so that $P'_{(\lambda)}v=A_{(\lambda)}v+B_{(\lambda)}V$ and $0=C_{(\lambda)}v+D_{(\lambda)}V$. This shows that $V$ is unique 
if and only if $D_{(\lambda)}$ is invertible. 
\par\smallskip
b) follows from the previous equations.
\par\smallskip
c) follows from b).
\end{pf}
\begin{rem}
As we wrote in the previous proof, in \cite{Dinar} it was shown that the uniqueness of the extension is equivalent to the invertibility of $D_{(\lambda)}$. Moreover, in the same paper it was proved that these conditions are equivalent to item b), and that its meaning is that the bi-Hamiltonian reduction amounts to a Dirac reduction. 
In the particular case of the Drinfeld-Sokolov reduction, this is related to the paper \cite{BFOFW90}, where the DS reduction is shown to be a Dirac reduction. For the purposes of our paper, it is more convenient to start with the hypothesis that the kernels of the Poisson tensors have trivial intersection, and the most important result in Proposition \ref{prop:schur} is the identity (\ref{decomp-poisson}). We will use it in Section \ref{sec:centralinvariants}
to show that, in the DS case, the $\lambda$-roots of the characteristic polynomials of the reduced and unreduced bi-Hamiltonian structures coincide. 
\end{rem}

We end this section by noticing that (\ref{decomp-poisson}) entails a general result in linear algebra, known as Schur determinant identity:
\begin{equation}
\label{Schur-determinant}
\det\begin{pmatrix} A_{(\lambda)} & B_{(\lambda)}\\ C_{(\lambda)} & D_{(\lambda)} \end{pmatrix}=
\det\left(A_{(\lambda)}-B_{(\lambda)}D_{(\lambda)}^{-1}C_{(\lambda)}\right)\,\det D_{(\lambda)}.
\end{equation}
In our setting, it means that the determinants of (the matrices representing) $P_{(\lambda)}$ and $P'_{(\lambda)}$ are related by 
\begin{equation}
\label{equal-determinant}
\det P_{(\lambda)}=(\det D_{(\lambda)})(\det P'_{(\lambda)}).
\end{equation}

\section{Reduction of exact bi-Hamiltonian manifolds}
\label{sec:exact-reduction}
A bi-Hamiltonian manifold $(\CM,P_1,P_2)$ is said to be {\em exact} if there exists a vector field $Z$ on $\CM$ such that $L_Z P_1=0$ and $L_Z P_2=P_1$. 
In terms of the corresponding Poisson brackets, this means that
\begin{equation}
Z\{F,G\}_1-\{ZF,G\}_1-\{F,ZG\}_1=0,\qquad Z\{F,G\}_2-\{ZF,G\}_2-\{F,ZG\}_2=\{F,G\}_1
\end{equation}
for all functions $F$, $G$ on $\CM$. The vector field $Z$ is called the {\em Liouville vector field\/} of the exact bi-Hamiltonian manifold. 
Exact bi-Hamiltonian manifolds have been studied in, e.g., \cite{Dorfman-book,Sergyeyev04}.

We now show that under a mild assumption on $Z$, the reduced bi-Hamiltonian manifold  is exact too. Combining this with a result of \cite{FL12}, stating that the central invariants of an exact semisimple Poisson pencil are constant, we obtain a criterion for proving the constancy of the central invariants of the reduced pencil,
provided that it admits a semisimple dispersionless limit.
This in particular applies to the case of Drinfeld-Sokolov considered in \cite{DLZ08}, for which we obtain an alternative proof of the constancy of the central invariants 
(relying on the fact, proved in \cite{DLZ08}, that the reduced DS pencil has a semisimple dispersionless limit --- see next section).

\begin{prop} 
Suppose $(\CM,P_1,P_2)$ to be an exact bi-Hamiltonian manifold, and the Liouville vector field $Z$ to be tangent to the symplectic leaf $\CS$. Then 
\begin{enumerate}
\item[a)] The restriction of $Z$ to $\CS$ projects onto a vector field $Z'$ on the quotient manifold $\CN$.
\item[b)] The reduced bi-Hamiltonian manifold $(\CN,P'_1,P'_2)$ is exact, with Liouville vector field $Z'$. 
\end{enumerate}
\end{prop}

\begin{pf}
a) We denote with $Z^{(\CS)}$ the restriction of $Z$ to $\CS$, and we recall that $Z^{(\CS)}$ can be projected along the (integral leaves of the) distribution $E$ if and only if $L_{Z^{(\CS)}}(E)\subset E$. This inclusion is a consequence of the fact that $Z$ is tangent to $\CS$ and that $L_{Z}(D)\subset D$. The last assertion can be checked as follows. If
$F$ is a Casimir of $P_1$, then 
\begin{equation}
L_Z(P_2dF)=\left(L_Z P_2\right)dF+P_2 d\left(L_Z F\right)=P_1 dF+P_2 d\left(L_Z F\right)=P_2 d\left(L_Z F\right)\in D,
\end{equation}
since it is easily seen that $L_Z F$ is also a Casimir of $P_1$. Indeed, $P_1 d\left(L_Z F\right)=P_1 \left(L_Z dF\right)=L_Z\left(P_1 dF\right)-\left(L_Z P_1\right)dF=0$. Hence there exists a vector field $Z'$ on $\CN$ such that $(Z'f)\circ\pi=Z^{(\CS)}(f\circ\pi)$ for all functions $f$ on $\CN$.
\par\smallskip
b) Let $f,g$ be functions on $\CN$, and let $F,G$ be extensions on $\CM$ as explained in Section \ref{sec:bi-Ham-reduction}. First, we notice that $ZF$ is an extension of $(Z'f)\circ\pi$ and that its differential vanishes on the tangent vectors in $D_{\bf w}$, for all ${\bf w}\in\CS$. Indeed, $(Z'f)(\pi({\bf w}))=\left(Z^{(\CS)}(f\circ\pi)\right)({\bf w})=(ZF)({\bf w})$ for all ${\bf w}\in\CS$, meaning that $ZF$ extends $(Z'f)\circ\pi$. Moreover, if $Y$ is a vector field in $D$, then $L_Y(ZF)=(L_Y Z)F+Z(L_Y F)=0$ since $L_Y Z=-L_Z Y$ is also in $D$.\\
Hence, for $a=1,2$, we have that
\begin{equation}
\begin{array}{l}
\left(Z'\{f,g\}'_a\right)(\pi({\bf w}))-\{Z'f,g\}'_a(\pi({\bf w}))-\{f,Z'g\}'_a(\pi({\bf w}))\\
\qquad=\left(Z^{(\CS)}(\{f,g\}'_a\circ\pi)\right)({\bf w})-\{ZF,G\}_a({\bf w})-\{F,ZG\}_a({\bf w})\\
\qquad=\left(Z\{F,G\}_a\right)({\bf w})-\{ZF,G\}_a({\bf w})-\{F,ZG\}_a({\bf w}),
\end{array}
\end{equation}
so that 
\begin{equation}
Z'\{f,g\}'_1-\{Z'f,g\}'_1-\{f,Z'g\}'_1=0,\qquad Z'\{f,g\}'_2-\{Z'f,g\}'_2-\{f,Z'g\}'_2=\{f,g\}'_1.
\end{equation}
\end{pf}

As we will see in the next section, an important case where this result can be applied is the Drinfeld-Sokolov (DS) reduction.

\section{Drinfeld-Sokolov structures}
\label{sec:DS-structures}

In this section we recall the bi-Hamiltonian structure related to the (untwisted) Drinfeld-Sokolov (DS) hierarchies \cite{DS85}, and the corresponding result about its central invariants \cite{DLZ08}. We will not use the original DS reduction, but an alternative (although equivalent) procedure
based on the general reduction theorem for bi-Hamiltonian manifolds discussed in Section \ref{sec:bi-Ham-reduction}. We refer to \cite{H78} for the basic facts concerning simple Lie algebras (see also the Appendix).

To obtain the (reduced) DS Poisson pair, the starting point is the loop algebra $\CM$ of functions $w(x)$ from the unit circle $S^1$ to a simple finite-dimensional Lie algebra 
$\mathfrak g$. We choose a Cartan subalgebra 
$\mathfrak h$ of $\mathfrak g$, we consider the corresponding principal gradation, and we select Chevalley (sometimes also called Weyl) generators 
$\{X_i,H_i,Y_i\}_{i=1,\dots,n}$ 
(where $n$ is the rank of $\mathfrak g$) with degrees 1,0, and -1, respectively.
One can
identify the cotangent spaces and the tangent spaces at every point $w\in\CM$ with $\CM$ itself, by means of the $\mbox{ad}$-invariant bilinear form 
$(w_1,w_2)=\int_{S^1}\langle w_1(x),w_2(x)\rangle_{\mathfrak g}\,dx$, where $\langle \cdot,\cdot\rangle_{\mathfrak g}$ is the normalized Killing form. 
Then one introduces the Poisson pair
\begin{equation}
\label{DS-pb}
\left(P_1\right)_w v=\epsilon^{-1}[v,A],\qquad \left(P_2\right)_w v=\epsilon^{-1}[v,w]+v_x,
\end{equation}
where $A\in\mathfrak g$ 
is an element of maximal degree. It is easily checked, see (\ref{DS-pb-delta}), that the Poisson pair (\ref{DS-pb}) has the form (\ref{bi-poisson}).
Moreover, the symplectic leaves of $P_1$ are affine subspaces over the vector space of maps from $S^1$ to $\ker(\ad A)^\perp$, 
where $\ker(\ad A)$ is the isotropy algebra of $A$ and the orthogonal is taken with respect to $\langle \cdot,\cdot\rangle_{\mathfrak g}$. An explicit description of $\ker(\ad A)^\perp$ is provided in the Appendix.

We choose the symplectic leaf $\CS$
containing the element $I=\sum_{i=1}^n Y_i$. Then the reduced bi-Hamiltonian manifold 
$\CN$ turns out to be parametrized by $n$ scalar-valued functions, and the reduced Poisson pair coincides \cite{CP92,P} with the DS one. We refer to \cite{Dinar} for an extension of this result to the so-called generalized DS structures (an example will be discussed in Section \ref{sec:beyondDS}). 

\begin{example}
\label{exa:kdv-red}
Let us consider the simplest case, ${\mathfrak g}={\mathfrak{sl}}(2)$, leading to the bi-Hamiltonian structure of the KdV hierarchy. The normalized Killing form is simply the trace of the product, and we choose the Chevalley generators
\begin{equation}
X=\begin{pmatrix}0 & 1\\\ 0 & 0 \end{pmatrix},\quad 
H=\begin{pmatrix}1 & 0\\ 0 & -1\end{pmatrix},\quad
Y=\begin{pmatrix}0 & 0\\ 1 & 0 \end{pmatrix}.
\end{equation}
Hence 
\begin{equation}
w=\begin{pmatrix}\frac12 w^2 & w^1\\ w^3 & -\frac 12 w^2 \end{pmatrix},\quad A=\begin{pmatrix}0 & 1\\\ 0 & 0 \end{pmatrix},\quad
I=\begin{pmatrix}0 & 0\\ 1 & 0 \end{pmatrix},
\end{equation}
so that the generic element $w$ of the symplectic leaf $\CS$ and the vectors of the distribution $D$ (at $w$) are given by
\begin{equation}
w=\begin{pmatrix}\frac12 w^2 & w^1\\ 1 & -\frac12 w^2 \end{pmatrix},\qquad \begin{pmatrix} \epsilon^{-1}k & k_x-\epsilon^{-1}w^2 k\\\ 0 & -\epsilon^{-1}k \end{pmatrix},
\end{equation}
where $k:S^1\to\mathbb R$ is any function. In this particular case, $D\subset T\CS$ and so $E=D$. Then it is easily shown that the projection $\pi:\CS\to\CN$ is given by 
$u=\pi(w^1,w^2)=w^1+\frac14\left(w^2\right)^2-\frac12\epsilon {w^2}_x$
and that the reduced Poisson pair turns out to be
\begin{equation}
\label{DS-pp-red}
\left(P_1'\right)_u=2\partial_x,\qquad \left(P_2'\right)_u=u_x+2u\partial_x-\frac12\epsilon^{2}\partial_{xxx},
\end{equation}
that is, the one associated to the Poisson brackets (\ref{kdv-pb-red}).
\end{example}

In the DS case, the choice of a transversal submanifold $\CQ$ corresponds to the choice of a ``canonical form" of the matrix Lax operator \cite{DS85}. For example, in the KdV case, we can take 
\begin{equation}
\label{kdv-Q}
\CQ=\left\{\begin{pmatrix} 0 & u \\ 1 & 0\end{pmatrix}\mid u\in C^\infty(S^1,\mathbb{R})\right\}.
\end{equation}

\begin{example}
\label{exa:so5}
To illustrate the use of a transversal submanifold in a more complicated example, we recall the case of $\mathfrak{so}(5)$, already treated in \cite{CP92}. Notice however that the choice of the Chevalley generators is different here --- it is the same as in \cite{DS85}. In particular, $\mathfrak{so}(5)$ is the Lie algebra of $5\times 5$ matrices such that $wS=-Sw^T$, where $S=\diag(1,-1,-1,-1,1)$. A set of Chevalley generators is given by
\begin{equation}
\begin{array}{lll}
X_1=e_{21}+e_{54},\quad & Y_1=X_1^T,\quad & H_1=-e_{11}+e_{22}-e_{44}+e_{55},\\
X_2=e_{32}+e_{43},\quad & Y_2=2X_2^T,\quad & H_2=2(e_{44}-e_{22}),
\end{array}
\end{equation}
where $e_{ij}$ is the matrix whose unique nonzero entry is $(i,j)$, which is 1. The principal gradation is completed by the 1-dimensional subspaces with degrees 2,-2,3,-3, generated respectively by
$$
X_3=-e_{31}+e_{53},\quad Y_3=-2X_3^T,\quad X_4=e_{41}+e_{52},\quad Y_4=4X_4^T.
$$
The normalized Killing form is $\langle w_1,w_2\rangle_{\mathfrak{so}(5)}=\frac{1}{2}\tr(w_1 w_2)$. To define the Poisson pair (\ref{DS-pb}), we can choose $A=X_4$. Its isotropy algebra is spanned by $X_1$, $X_2$, $X_3$, $X_4$, $H_1$, $Y_1$, so that the symplectic leaf $\CS$ is the space of maps 
$$
x\mapsto w^1(x)X_4+w^2(x)X_2+w^3(x)X_3+w^4(x)(H_1+H_2)+Y_1+Y_2.
$$
A possible choice for the transversal submanifold is the subset $\CQ\subset\CS$ of the maps with $w^3(x)=w^4(x)=0$. Then the reduced structures turn out to be given by 
\begin{equation}
P_1'=\begin{pmatrix} \frac12\epsilon^2 \partial_x^3-2w^2\partial_x-w^2_x & 2\partial_x \\ 2\partial_x & 0\end{pmatrix},\qquad
P_2'=\begin{pmatrix} (P_2')_{11} & (P_2')_{12} \\ (P_2')_{21} & (P_2')_{22}\end{pmatrix},
\end{equation}
with
\begin{equation*}
\begin{aligned}
&(P_2')_{11}=\textstyle{-\frac{1}{16}\epsilon^6 \partial^{7}_x+\frac{1}{2}\epsilon^4 w^2 \partial^{5}_x +\frac{5}{4}\epsilon^4 w^2_x\partial^{4}_x 
+ \epsilon^2\left(\frac{1}{2} w^1-(w^2)^2+2 \epsilon^2 w^2_{xx}\right)\partial^{3}_x}\\
&
\textstyle{+\epsilon^2\left(\frac{3}{4} w^1_x -3 w^2 w^2_x+\frac{7}{4} \epsilon^2 w^2_{xxx}\right)\partial^2_x 
+\left(-2 w^1 w^2+\epsilon^2(-\frac{3}{4} (w^2_x)^2+\frac{3}{4} w^1_{xx}-2 w^2 w^2_{xx})+\frac{3}{4}\epsilon^4 w^2_{xxxx} \right)\partial_x}\\
&
\textstyle{-w^1_x w^2 - w^1 w^2_x+\epsilon^2(-\frac{1}{4}w^2_x w^2_{xx}+\frac{1}{4} w^1_{xxx}-\frac{1}{2} w^2 w^2_{xxx}) +\frac{1}{8} \epsilon^4 w^2_{xxxxx}}\\
\noalign{\smallskip}
&(P_2')_{12}=\textstyle{-\frac{1}{4} \epsilon^4\partial^{5}_x+\epsilon^2 w^2\partial^{3}_x +\frac{1}{2} \epsilon^2 w^2_x\partial^2_x +2 w^1 \partial_x+\frac{1}{2} w^1_x}\\
\noalign{\smallskip}
&(P_2')_{21}=\textstyle{-\frac{1}{4} \epsilon^4\partial^{5}_x+\epsilon^2 w^2\partial^{3}_x +\frac{5}{2} \epsilon^2 w^2_x\partial^2_x +
\left(2 w^1 +2 \epsilon^2 w^2_{xx}\right)  \partial_x+\frac{3}{2} w^1_x+\frac{1}{2} \epsilon^2 w^2_{xxx}}\\
\noalign{\smallskip}
&(P_2')_{22}=\textstyle{-\frac{5}{4}\epsilon^2 \partial^{3}_x+w^2 \partial_x+\frac{1}{2} w^2_x}
\end{aligned}
\end{equation*}
\end{example}

\bigskip\par\noindent

The following important results on the DS structures were shown in \cite{DLZ08}:
\begin{enumerate}
\item the (reduced) DS Poisson pair has the form (\ref{poissondisp}), and its dispersionless part is given by a semisimple flat pencil of metrics (described in \cite{S93}) on the orbit space ${\mathfrak h}/W$ of the Weyl group $W$ of the Lie algebra ${\mathfrak g}$; 
\item its central invariants are constant; more precisely, they are given by 
\begin{equation}
\label{central-inv-DS}
c_i=\frac1{48}\langle H_i,H_i\rangle_{\mathfrak g}. 
\end{equation}
\end{enumerate}

From our point of view, the constancy of the central invariants can be proved with the help of the results of Section \ref{sec:exact-reduction} 
and item 1 above. Indeed, 
it is easily seen that the Poisson pair (\ref{DS-pb}) 
is exact, with the Liouville vector field simply given by $Z_w=A$. This vector field is tangent to every symplectic leaf of $P_1$, since $A=[h,A]$ for a suitable element $h\in\mathfrak h$. (Notice that this property is true for every root vector in $\mathfrak g$). Hence we can conclude that the reduced bi-Hamiltonian manifold is exact too. Thanks to the results of \cite{FL12}, the corresponding central invariants are constant. 

\begin{example}
\label{exa:kdv-exact}
Using the same notations as in Example \ref{exa:kdv-red}, the Liouville vector field $Z$ on the (unreduced) manifold $\CM$ is 
$Z_w=
\left(\begin{smallmatrix}0 & 1\\ 0 & 0 \end{smallmatrix}\right)$. 
It is immediate to verify that it is tangent to $\CS$, can be projected to $\CN$, and its projection is $Z'_u=1$. It is also easily checked that $Z'$ is a Liouville vector field for the 
(reduced) manifold $\CN$.
\end{example}

\section{Invariants of DS structures}
\label{sec:centralinvariants} 

In this section we use the same notations as in the previous one, and we show the equality of the characteristic polynomials (\ref{charpoly}) of the reduced and unreduced DS bi-Hamiltonian structures. 
In Section \ref{sec:beyondDS} we will discuss the examples of a generalized DS hierarchy and of the Camassa-Holm equation, showing that the equality
holds in more general contexts than the DS reduction. 

Let us introduce a basis $\{e^l\}$ of the Lie algebra $\mathfrak g$, with $l=1,\dots,N$. If 
\begin{equation}
[e^i,e^j]=\sum_{l=1}^N c^{ij}_l e^l,\qquad g^{ij}=\langle e^i,e^j\rangle_{\mathfrak g},
\end{equation} 
then we can write the Poisson pencil (\ref{DS-pb}) in the form, analogous to (\ref{bi-poisson}), 
\begin{equation}
\label{DS-pb-delta}
\{w^i(x),w^j(y)\}_{(\lambda)}=-\epsilon^{-1} \sum_{l=1}^N c^{ij}_l (w^l-\lambda A^l) \delta(x-y)-g^{ij}\delta'(x-y),
\end{equation}
where $v^l=\langle v,e^l\rangle_{\mathfrak g}$ for any $v	\in{\mathfrak g}$. According to (\ref{poisson-symbol}) and (\ref{charpoly}), 
we associate to the DS (unreduced) Poisson pencil (\ref{DS-pb-delta}) on $\CM$ the $\lambda$-polynomial
\begin{equation}
\label{charpoly-up}
{\mathcal R}_\CM(p,\lambda;
w
)=\det\left(-\sum_{l=1}^N c^{ij}_{l} (w^l-\lambda A^l)-g^{ij}p\right).
\end{equation}
Note that the above polynomial is explicitly written in terms of Lie algebra objects; this is in general not so for the characteristic polynomial of the reduced pencil. However, we have:
\begin{theorem}
\label{theorem:equal-charpoly} The $\lambda$-roots of the characteristic polynomial are preserved by the bi-Hamiltonian reduction. More precisely, if $w\in\CS$, then 
\begin{equation}
\label{charpoly-conjecture}
{\mathcal R}_\CN(p,\lambda;\pi_0(
w
))=
F\,{\mathcal R}_\CM(p,\lambda;
w
),
\end{equation}
where $\pi_0:\CS\to\CN$ is the $\epsilon$-independent part of the projection $\pi:\CS\to\CN$ on the reduced bi-Hamiltonian 
manifold, ${\mathcal R}_\CN
$ is the characteristic polynomial of the (reduced) DS Poisson pair on $\CN$, and 
 $F$ is a non vanishing function, independent of $\lambda$.
\end{theorem}

In terms of the reduced Poisson pair on a transversal manifold $\CQ$,  Theorem \ref{theorem:equal-charpoly} admits the following equivalent form:
\begin{theorem}
\label{theorem:equal-charpoly-Q}
Let $\CQ$ be a transversal submanifold of the DS symplectic leaf $\CS$, and ${\mathcal R}_\CQ$ the characteristic polynomial associated with the reduced Poisson pair 
on $\CQ$. Then, for all $w\in\CQ$, 
\begin{equation}
\label{charpoly-conjecture-Q}
{\mathcal R}_\CQ(p,\lambda;{
w})=
F_\CQ\,{\mathcal R}_\CM(p,\lambda;{
w}),
\end{equation}
where 
$F_\CQ$ is a non vanishing function, independent of $\lambda$.
\end{theorem}

In the examples  we considered, $F$ and $F_\CQ$ do not even depend on $p$ and on $w$, i.e., they are constant depending only the choice of the basis $\{e^l\}$. We guess that this is a general fact.
\par\smallskip\noindent
{\bf Proof of Theorem \ref{theorem:equal-charpoly-Q}.} 
Proposition \ref{prop:schur} can be applied in our case, since the condition $(\ker P_1)_w
\cap (\ker P_2)_w
=\{0\}$ for all $w
\in\CQ$ is satisfied. This was shown in \cite[Proposition 3.2]{CFMP98}. (In that paper, the case $\mathfrak g={\mathfrak{sl}}(n)$ was considered, 
but the proof is the same for every simple Lie algebra. For the reader's convenience, we explicitly provide such general proof in Appendix A.2).
Then formula (\ref{decomp-poisson}) holds, and for the corresponding power series in $p$, see (\ref{poisson-symbol}), we obtain
\begin{equation}
\label{decomp-poisson-symb}
\pi_{(\lambda)}=\begin{pmatrix}  \pi'_{(\lambda)} & \beta_{(\lambda)}\\ 0 & \delta_{(\lambda)}\end{pmatrix}
\begin{pmatrix} \mbox{Id} & 0\\ \delta_{(\lambda)}^{-1}\gamma_{(\lambda)} & \mbox{Id}
\end{pmatrix}.
\end{equation}
Indeed, let $P$ and $R$ be $N\times N$ matrices of differential operators with entries 
$$
P^{ij}=\sum_{k\ge -1}\epsilon^k\sum_{l=0}^{k+1}A^{ij}_{k,l}
\partial_x^{k-l+1},\qquad
R^{ij}=\sum_{k\ge -1}\epsilon^k\sum_{l=0}^{k+1}B^{ij}_{k,l}
\partial_x^{k-l+1},
$$
where $A^{ij}_{k,l}({\bf w},{\bf w}_x,\dots,{\bf w}_{(l)})$ and $B^{ij}_{k,l}({\bf w},{\bf w}_x,\dots,{\bf w}_{(l)})$ are differential polynomials of degree $l$, as in (\ref{poisson}). 
If $\pi$ and $\rho$ are the corresponding matrices with entries
$$
\pi^{ij}=\sum_{k\ge -1}A^{ij}_{k,0}({\bf w})p^{k+1},\qquad
\rho^{ij}=\sum_{k\ge -1}B^{ij}_{k,0}({\bf w})p^{k+1},
$$
then it is easily shown that the matrix associated to $PR$ is $\pi\rho$. This implies that if $P$ is invertible (as a matrix differential operator), then
also $\pi$ is invertible (as a matrix), and the matrix associated to $P^{-1}$ is $\pi^{-1}$.

Therefore, from (\ref{decomp-poisson-symb}) it follows that
\begin{equation}
\label{equal-determinant-symb}
\det \pi_{(\lambda)}=(\det \delta_{(\lambda)})(\det \pi'_{(\lambda)}),
\end{equation}
where $\det \delta_{(\lambda)}$ never vanishes. Hence 
\begin{equation}
\label{equal-determinant-R}
{\mathcal R}_\CM(p,\lambda;{
w})=
(\det \delta_{(\lambda)})
{\mathcal R}_\CQ(p,\lambda;{
w})
\end{equation}
for all $w
\in\CQ$, and we are left with showing that $\det \delta_{(\lambda)}$ does not depend on $\lambda$. This follows from the fact that the degree in $\lambda$ of both 
${\mathcal R}_\CM(p,\lambda;{
w})$ and ${\mathcal R}_\CQ(p,\lambda;{
w})$ is equal to the rank $n$ of the Lie algebra $\mathfrak g$. As far as ${\mathcal R}_\CQ(p,\lambda;{
w})$ is concerned, it is a general property of $n$-field semisimple flat pencils of metrics. The result concerning  ${\mathcal R}_\CM(p,\lambda;{
w})$ is proved in the Appendix and makes use of Lemma \ref{lemmaA3}, which in turn follows from \cite{MRV16} (see also \cite{MRV17}).
\par\hfill$\square$\par\smallskip

\begin{example}
\label{exa:kdv-conj}
Consider once more the case ${\mathfrak g}={\mathfrak{sl}}(2)$, as in Example \ref{exa:kdv-red}. In the coordinates $(w^1,w^2,w^3)$ (that is, using the basis $e^1=Y$, $e^2=H$, $e^3=X$) the matrix polynomial (\ref{poisson-symbol}) associated to the Poisson pencil (\ref{DS-pb-delta}) is given by (\ref{sl2symbol-sec2}), 
see Example \ref{exa:kdv-sec2}, so that the corresponding characteristic polynomial, evaluated at the points of the symplectic leaf $\CS$, is 
\begin{equation}
{\mathcal R}_\CM(p,\lambda;w^1,w^2)=2p^3-8\left(w^1+\frac14\left(w^2\right)^2-\lambda\right)p. 
\end{equation}
Since $\pi_0:(w^1,w^2)\mapsto u=w^1+\frac14\left(w^2\right)^2$ 
and ${\mathcal R}_\CN(p,\lambda;u)$ is given by (\ref{kdv-charpoly}), we see that (\ref{charpoly-conjecture}) is satisfied with $F=\frac14$.
\end{example}

\begin{example}
In the case of $\mathfrak{so}(5)$, see Example \ref{exa:so5}, we have that the characteristic polynomial ${\mathcal R}_\CM$, evaluated at the points of the 
transversal manifolds $\CQ$, is
\begin{equation}
256\,{\mathcal R}_\CQ(p,\lambda;w^1,w^2)=4 p^2 \left(-32 \lambda +p^4-8 p^2 w^2+32 w^1+
16 \left(w^2\right)^2\right) \left(8 \lambda +p^4-4 p^2 w^2-8 w^1\right).
\end{equation}
\end{example}

\section{Beyond DS structures}
\label{sec:beyondDS} 
In this final section we discuss two examples, in order to support the claim that Theorem \ref{theorem:equal-charpoly-Q} holds in more general contexts than the DS reduction. 
The first one is related to the $W_3^{(2)}$ conformal algebra of \cite{MO91}, i.e., to the so called fractional KdV hierarchy  ${\mathfrak{sl}}_3^{(2)}$ discussed in \cite{BDHM93}. It has been already treated from the bi-Hamiltonian point of view in \cite{CFMP97}, and belongs to the class of generalized DS bi-Hamiltonian structures (see \cite{BDHM93,Dinar}). The second example is the Poisson pair of the Camassa-Holm (CH) 
equation \cite{CH93}.

We consider again the Poisson tensors (\ref{DS-pb}), with ${\mathfrak g}={\mathfrak{sl}}_3$ and
\begin{equation}
A=\begin{pmatrix}0 & 1 & 0\\\ 0 & 0 & 1\\ 0 & 0 & 0\end{pmatrix}.
\end{equation}
We choose the symplectic leaf $\CS$ of $P_1$ containing the point 
\begin{equation}
I=\begin{pmatrix}0 & 0 & 0\\\ 0 & 0 & 0\\ 1 & 0 & 0\end{pmatrix}.
\end{equation}
It can be checked that its elements are 
\begin{equation}
\begin{pmatrix} p_0 & u_1 & u_3\\\ p_1 & u_0-p_0 & u_2\\ 1 & -p_1 & -u_0\end{pmatrix},
\end{equation}
and that a transversal submanifold $\CQ$ is given by $p_0=p_1=0$. The reduced Poisson pencil is\footnote{The formulae written in \cite{CFMP97} contain a misprint. Those given in the present paper are correct.}
\begin{equation}
P_{(\lambda)}'=\left(\begin{array}{cccc}
\frac{2}{3}\partial_x   &  \epsilon^{-1}(u_1-\lambda)  & \epsilon^{-1}(-u_2+\lambda) &
-\frac{1}{3}\left(\epsilon\partial_x^2 + u_0\partial_x +{u_0}_x\right)\\
 \ast & 0 &(P_{(\lambda)}')_{23}&  2(u_1-\lambda)\partial_x+{u_1}_x+2\epsilon^{-1}u_0(u_1-\lambda)\\
 \ast & \ast & 0 & (u_2-\lambda)\partial_x+{u_2}_x-2\epsilon^{-1}u_0(u_2-\lambda)\\
 \ast & \ast & \ast &
(P_{(\lambda)}')_{44}
\end{array}\right),
\end{equation}
where
\begin{equation}
\begin{array}{l}
(P_{(\lambda)}')_{23}=\epsilon\partial_x^2 + 3 u_0 \partial_x + 2 {u_0}_x+ \epsilon^{-1}\left(2 {u_0}^2-u_3\right) \\ 
(P_{(\lambda)}')_{44}=-\frac{2}{3} \epsilon^2\partial_x^3 - \frac{4}{3}\epsilon {u_0}_x\partial_x+ 2\left(u_3+ \frac{1}{3} {u_0}^2\right)\partial_x
- \frac{2}{3}\epsilon {u_0}_{xx}+ \frac{2}{3} u_0 {u_0}_x+  {u_3}_x.
\end{array}
\end{equation}
Notice that the reduced Poisson pair, like the unreduced one, is exact but it does not admit a dispersionless limit. However, one can check
that the characteristic polynomial (\ref{charpoly-up}), computed at the points of $\CQ$, coincides with 
$-3{\mathcal R}_\CQ(p,\lambda;u_0,u_1,u_2,u_3)$, so that Theorem \ref{theorem:equal-charpoly-Q} holds in this case with $F_\CQ=-\frac13$.

Now we pass to the CH case. First of all, we need a brief summary of the fact \cite{LP04} that the CH Poisson pair can also be obtained by reduction from a Poisson pair on loop algebras, very similar to (\ref{DS-pb}).

Consider again the loop algebra $\CM$ of functions $w(x)$ from 
$S^1$ to $\mathfrak{sl}(2)$, endowed now with the Poisson pair
\begin{equation}
\label{ch-pb}
\left(P_1\right)_w v=\epsilon^{-1}[v,w],\qquad \left(P_2\right)_w v=\epsilon^{-1}[v, A]+v_x,
\end{equation}
where $ A=X+Y=\left(\begin{smallmatrix}0 & 1\\ 1 & 0 \end{smallmatrix}\right)$. Also in this case, the Poisson pair has the form (\ref{bi-poisson}). We choose the symplectic leaf of $P_1$ given by $\CS=\{w\in\CM\mid \det w=0\}$. Since the quotient $\CN$ is not easy to parametrize, we introduce the transversal submanifold 
\begin{equation}
\label{ch-Q}
\CQ=\left\{\begin{pmatrix} 0 & u \\ 0 & 0\end{pmatrix}\mid u\in C^\infty(S^1,\mathbb{R}), u(x)\ne 0\,\forall x\in S^1\right\}.
\end{equation}
The reduced Poisson brackets turn out to be the ones of the Camassa-Holm hierarchy, that is,
\begin{equation}
\label{ch-pb-red}
\{u(x),u(y)\}_1=-u_x\delta(x-y)-2u\delta'(x-y),\qquad \{u(x),u(y)\}_2=-2\delta'(x-y)+\frac12\epsilon^2\delta'''(x-y),
\end{equation}
to be compared with (\ref{kdv-pb-red}). The corresponding characteristic polynomial is thus 
\begin{equation}
{\mathcal R}_\CQ(p,\lambda;u)=
-2p+\frac12 p^3+2\lambda u p.
\end{equation} 
Let us compute now the characteristic polynomial of the (unreduced) CH structure. In the same coordinates $(w^1,w^2,w^3)$ used in 
Example \ref{exa:kdv-red}, the matrix polynomial (\ref{poisson-symbol}) associated to the Poisson pencil (\ref{ch-pb}) is
\begin{equation}
\begin{pmatrix} 0 & 2(\lambda w^1-1) & -\lambda w^2-p\\ -2(\lambda w^1-1) & -2p & 2(\lambda w^3-1) \\ \lambda w^2-p & -2(\lambda w^3-1) & 0\end{pmatrix},
\end{equation}
so that the corresponding characteristic polynomial is 
\begin{equation}
{\mathcal R}_\CM(p,\lambda;w^1,w^2,w^3)=2p^3-2p\left[\left(4w^1w^3+\left(w^2\right)^2\right)\lambda^2-4(w^1+w^3)\lambda+4\right]. 
\end{equation}
Since $\CQ$ is defined by $w^1=u$, $w^2=w^3=0$, we have that 
\begin{equation}
{\mathcal R}_\CQ(p,\lambda;u)=\frac14{\mathcal R}_\CM(p,\lambda;u,0,0). 
\end{equation}
We can conclude that also in the CH case a relation of the form (\ref{charpoly-conjecture-Q}) holds.

\section*{Appendix}

In this appendix we collect some facts, concerning simple Lie algebras, used in the paper.

\subsection*{A.1 The symplectic leaves of $P_1$}

Let $\mf{g}$ be a simple Lie algebra of rank $n$ over $\bb{C}$. Fix a Cartan subalgebra $\mf{h}$, denote by $\Delta\subset\mf{h}^\ast$ the root space of $\mf{g}$, and for every $\alpha\in\Delta$ denote by $E_\alpha\in\mf{g}$ a nonzero vector in the corresponding root space, so that $\mf{g}$ decomposes as
$$\mf{g}=\mf{h}\oplus\bigoplus_{\alpha\in\Delta}\bb{C}E_\alpha.$$
Fix a base of simple roots $\Pi=\{\alpha_1,\dots,\alpha_n\}\subset\Delta$, and denote by $\Delta_+$ (resp.\ $\Delta_-$) the corresponding set of positive (resp.\ negative) roots. Denote $\mf{n}_{\pm}=\bigoplus_{\alpha\in\Delta_{\pm}}\bb{C}E_\alpha$, from which the decomposition $\mf{g}=\mf{n}_-\oplus\mf{h}\oplus\mf{n}_+$ follows. Given  $\alpha\in\Delta$, denote by $\Ht(\alpha)$ the height of $\alpha$ relative to $\Pi$, and extend this to $\mf{g}$ by setting $\Ht(E_\alpha)=\Ht(\alpha)$, $\alpha\in\Delta$, 
and $\Ht(h)=0$, $h\in\mf{h}$. Let $h^\vee$ be the dual Coxeter number of $\mf{g}$, and introduce the normalized Killing form
$$
\langle x,y\rangle_{\mathfrak g}=\frac{1}{2h^\vee}\tr(\ad x\ad y), \qquad x,y\in\mf{g}.
$$
Let $\theta\in\Delta$ be the highest root of $\mf{g}$, and consider the corresponding root vector $E_\theta$ (this is a possible choice of the element $A$ entering definition of the Poisson structure $P_1$ in Section \ref{sec:DS-structures}). 
To describe the symplectic leaf
$$
\ker(\ad E_\theta)^\perp=\{x\in\mf{g}\mid \langle x,y\rangle_{\mathfrak g}=0\,\,\forall\, y\in\ker(\ad E_\theta)\}
$$
in more detail, we introduce the following gradation. Let $\nu:\mf{h}\to\mf{h}^\ast$ be the isomorphism of vector spaces given by 
$\langle \nu(x),y\rangle=\langle x,y\rangle_{\mathfrak g}$, $x,y\in\mf{h}$, and let $\theta^\vee=\nu^{-1}(\theta)\in\mf{h}$. Then $\mf{g}$ decomposes as follows,
$$\mf{g}=\bigoplus_{i=-2}^{2}\mf{g}^i,$$
where $\mf{g}^i=\{x\in\mf{g}\mid [\theta^\vee,x]=i\,x\}$. The subalgebra $\mf{g}^0$ is a reductive subalgebra which contains $\mf{h}$, while $\mf{g}^1\oplus\mf{g}^2$, 
 (resp. $\mf{g}^{-1}\oplus\mf{g}^{-2}$) is a nilpotent subalgebra contained in $\mf{n}_+$ (resp.\ $\mf{n}_-$). Moreover, it can be shown \cite{MR18} that $\dim\mf{g}^1=\dim\mf{g}^{-1}=2(h^\vee-2)$ and  $\mf{g}^{\pm2}=\bb{C}E_{\pm\theta}$. 

\begin{lem}
We have that 
$\ker(\ad E_\theta)^\perp=\bb{C}\theta^\vee\oplus\mf{g}^1\oplus\mf{g}^2$. In particular, $\dim\left(\ker(\ad E_\theta)^\perp\right)=2h^\vee-2$.
\end{lem}
\begin{pf}
We first describe the subspace $\ker(\ad E_\theta)$. It is clear that for $x\in\mf{n}_+$ then $x\in \ker(\ad E_\theta)$. For $x\in\mf{h}$, since $x\in\ker(\ad E_\theta)$ if and only if $\langle\theta,x\rangle=0$, and since $\langle\theta,\theta^\vee\rangle=2$, it follows that the Cartan subalgebra decomposes into the (orthogonal) direct sum $\mf{h}=\widetilde{\mf{h}}\oplus\bb{C}\theta^\vee$, where $\widetilde{\mf{h}}=\{x\in\mf{h}\mid \langle x,\theta\rangle=0\}$. If $x\in\mf{g}^0\cap \mf{n}_-$, then we have both $[x,E_\theta]\in\mf{g}^2=\bb{C}E_\theta$ and  $\Ht\left([x,E_\theta]\right)<\Ht(E_\theta)$, which implies $[x,E_\theta]=0$. For $x\in\mf{g}^{i}$, $i=-1,-2$, then $[E_{-\theta},[x,E_\theta]]=-[x,[E_\theta,E_{-\theta}]]-[E_\theta,[E_{-\theta},x]]=[x,\theta^\vee]=-i x$, so that $[x,E_\theta]=0$ if and only if $x=0$. Thus for $i=-1,-2$ we have 
$\mf{g}^i\cap\ker(\ad E_\theta)=\{0\}$, and we proved that
\beq\label{keradetheta}
\ker(\ad E_\theta )=\widetilde{\mf{g}}^0\oplus\mf{g}^1\oplus\mf{g}^2,
\eeq
where $\widetilde{\mf{g}}^0=(\mf{g}^0\cap\mf{n}_-)\oplus\widetilde{\mf{h}}\oplus(\mf{g}^0\cap\mf{n}_+)$. Note in particular that $\mf{g}^0=\widetilde{\mf{g}}^0\oplus \bb{C}\theta^\vee$, and that $\dim\ker(\ad E_\theta )=\dim \mf{g}^0-1+\dim\mf{g}^1+\dim \mf{g}^2=\dim \mf{g}^0+\dim\mf{g}^1$. From \eqref{keradetheta} and the properties of the Killing form it is easy to show that  $\bb{C}\theta^\vee\oplus\mf{g}^1\oplus\mf{g}^2\subset \ker(\ad E_\theta)^\perp$. It remains to prove that the latter inclusion is in fact an equality. Since $\dim\mf{g}=\dim\mf{g}^0+2\dim\mf{g}^1+2$ we have $\dim \left(\ker(\ad E_\theta)^\perp\right)=\dim\mf{g}-\dim \left(\ker(\ad E_\theta)\right)=(\dim\mf{g}^0+2\dim\mf{g}^1+2)-(\dim \mf{g}^0+\dim\mf{g}^1)=\dim\mf{g}^1+2$, from which the thesis follows. Moreover, since $\dim\mf{g}^1=2(h^\vee-2)$, we get $\dim \left(\ker(\ad E_\theta)^\perp\right)=2h^\vee-2.$
\end{pf}

We remark that the characterization of the symplectic leaf $\ker(\ad E_\theta)^\perp$ provided in the above lemma is not necessary for the results of present paper. We believe however that it is a nice piece of information, which easily follows from the results of \cite{MR18} and which to our knowledge has never appeared in the literature before. In Table 1
the dimension of $\ker(\ad E_\theta)^\perp$ is summarized for every finite-dimensional simple Lie algebra.

\begin{table}[H]
\label{tablecoxeter}
\normalsize
\caption{Coxeter numbers, dual Coxeter numbers and the dimension of 
$\ker(\ad E_\theta)^\perp$ for simple Lie algebras.}
\begin{center}
\begin{tabular}{|c||c|c|c|c|c|c|c|c|c|} 
\hline
$\mf{g}$ & $A_n$ & $B_n$  & $C_n$ & $D_n$ & $E_6$ & $E_7$ & $E_8$ & $F_4$ & $G_2$
\\
\hline
$h$ &  $n+1$ & $2n$ & $2n$ & $2n-2$ & $12$ & $18$ & $30$ & $12$ & $6$
\\
\hline
$h^\vee$ &  $n+1$ & $2n-1$ & $n+1$ & $2n-2$ & $12$ & $18$ & $30$ & $9$ & $4$
\\
\hline
$\dim\left(\ker(\ad E_\theta)^\perp\right)$ &  $2n$ & $4n-4$ & $2n$ & $4n-6$ & $22$ & $34$ & $58$ & $16$ & $6$
\\
\hline
\end{tabular}
\end{center}
\end{table}


\subsection*{A.2 The intersection of the kernels in the DS case}

Let us denote by $h$ the Coxeter number of $\mf g$ and by $N$ its dimension. Then $N=n(h+1)$. Let $\left\{X_i,Y_i,H_i\right\}_{i=1,\dots,n}$ 
be Chevalley generators of $\mf{g}$, satisfying the relations
$$[H_i,H_j]=0,\quad [H_i,X_j]=C_{ij}X_j,\quad [H_i,Y_j]=-C_{ij}Y_j,\quad [X_i,Y_j]=\delta_{ij}H_i,$$
where $C=(C_{ij})_{i,j=1,\dots,n}$ is the Cartan matrix of $\mf g$.  Define the principal gradation of $\mf g$ by setting $\deg X_i=-\deg Y_i=1$, for every $i=1,\dots,n$, and extending it uniquely to a gradation of the algebra by setting $\deg [x,y]=\deg x+\deg y$. Then we have
$$\mf{g}=\bigoplus_{i=1-h}^{h-1}\mf{g}_i,\qquad \text{where}\qquad \mf{g}_i=\left\{x\in\mf{g}\mid\deg x=i\right\}.$$
Note in particular that for every $i=1,\dots n$  we have $X_i\in\mf{g}_1$, $Y_i\in \mf{g}_{-1}$ and that $\dim \mf{g}_{h-1}=1$. 
The principal gradation is induced by the adjoint action of the element
$$\rho^\vee=\sum_{j,k=1}^n\left(C^{-1}\right)_{kj} H_j\in\mf{h},$$
so that we can equivalently write $\mf{g}_i=\left\{x\in\mf{g}\mid\ad_{\rho^\vee} x=ix\right\}$.

Denote $I=\sum_{i=1}^n Y_i$ the principal nilpotent element and let $A=E_\theta\in \mf{g}_{h-1}$. 
The other notations are as in Sections \ref{sec:DS-structures} and \ref{sec:centralinvariants}. 
We will show that the intersection between $(\ker P_1)_w$ and $(\ker P_2)_w$ is trivial for all $w\in\CS$ (more generally, for all 
$w\in I+\bigoplus_{i=0}^{h-1}\mf{g}_i$). Indeed, suppose that 
$v\in(\ker P_1)_w\cap(\ker P_2)_w$ and decompose it with respect to the principal gradation: $v=\sum_{i=1-h}^{h-1}v_i$. Since $v\in(\ker P_1)_w$, every $v_i$ commutes with $A$ and therefore $v_{1-h}=0$. Decomposing also $w=I+\sum_{j=0}^{h-1}w_j$, imposing that $v\in(\ker P_2)_w$, and considering the 
minimal degree element, we obtain that $[v_{2-h},I]=0$. It follows that $v_{2-h}$ commutes both with $A$ and $I$, so that $v_{2-h}=0$ (see~\cite{CP92}). In the same way, one proves that $v_i=0$ for all $i=3-h,\dots,h-1$. 


\subsection*{A.3 The $\lambda$-degree of the characteristic polynomial}

This part of the appendix is devoted to the proof that the degree in $\lambda$ of the characteristic polynomial ${\mathcal R}_\CM(p,\lambda;{
w})$ is equal to the rank $n$ of $\mathfrak g$, if $w$ belongs to the symplectic leaf $\CS$. First of all, we notice that 
$$
{\mathcal R}_\CM(p,\lambda;{
w})=\det\left(g^{ij}\right)\det\left(-p\, \mbox{Id}- \ad(w-\lambda A)\right),
$$
so that we just need to compute the degree of $\det\left(p\, \mbox{Id}+ \ad(w-\lambda A)\right)$ in the adjoint representation.

Denote $I=\sum_{i=1}^n Y_i$ the principal nilpotent element, let $E_\theta\in \mf{g}_{h-1}$ and denote  $\Lambda_\lambda=I-\lambda E_\theta=I-\lambda A$. 
We first compute the degree in $\lambda$ of the characteristic polynomial 
$$C_{\lambda}(p)=\det\left(p\, \mbox{Id}+ \ad\Lambda_\lambda\right).$$ 
It is known \cite{Ko78} that the element $\Lambda_1=I- E_\theta$ is a regular semisimple element of $\mf{g}$. The following easy lemma is a special case of \cite[Lemma 4.1]{MRV16}:
\begin{lem}\label{lemmaA3}
If $p$ is an eigenvalue for $\Lambda_1$, then $\lambda^{\frac{1}{h}}p$ is an eigenvalue for $\Lambda_\lambda.$
\end{lem}
\begin{proof}
Since $
I\in\mf{g}_{-1}$ and $E_\theta\in\mf{g}_{h-1}$ we have
$$\lambda^{\frac{1}{h}ad\rho^\vee}\Lambda_1=\lambda^{\frac{1}{h}ad\rho^\vee}
I-\lambda^{\frac{1}{h}ad\rho^\vee}E_\theta=\lambda^{-\frac{1}{h}}
I-\lambda^{\frac{h-1}{h}}E_{\theta}=\lambda^{-\frac{1}{h}}\Lambda_\lambda,$$
proving the lemma.
\end{proof}
Now   $\Lambda_1$ is semisimple and therefore diagonalizable, and its characteristic polynomial 
is of the form
$$C_1(p)=p^n\prod_{i=1}^{hn}(p-p_i),$$
for certain $p_i\in\bb{C}\setminus\{0\}$, $i=1,\dots, hn$.  In particular, the contribution $p^n$ in the polynomial above is due to the fact that $\Lambda_1$ is regular. 
Due to the previous lemma, the characteristic polynomial $C_\lambda(p)$ of $\Lambda_\lambda$   
is given by
$$C_\lambda(p)=p^n\prod_{i=1}^{hn}(p-\lambda^{\frac{1}{h}}p_i),$$
which is a polynomial of degree $\frac{nh}{h}=n$ with respect to $\lambda$. 

Now consider the general case, i.e., the element $\Lambda_\lambda+ w_+=I+w_+-\lambda A=w-\lambda A$, where $w\in\CS$ and 
$$w_+=\sum_{i=0}^{h-1} w_i , \qquad w_i\in \mf{g}_{i}
.$$
Then we obtain
\begin{align}
\lambda^{-\frac{1}{h}\ad \rho^\vee}(\Lambda_\lambda+ w_+)&=\lambda^{\frac{1}{h}}I-\lambda \lambda^{\frac{1-h}{h}}E_{\theta}+
\sum_{i=0}^{h-1} \lambda^{-\frac{i}{h}}w_i\notag \\
&=\lambda^{\frac{1}{h}}(\Lambda_1+\lambda^{-\frac{1}{h}}M(\lambda;w_+)),
\end{align}
where for fixed $w_+$ the function $M(\lambda;w_+)$ is a polynomial in $\lambda^{-\frac{1}{h}}$. It then follows that --- for a fixed value of $\lambda$ --- $p(\lambda)$ is an eigenvalue of $\Lambda_\lambda+w_+$ if and only if $\lambda^{\frac{1}{h}}p(\lambda)$ is an eigenvalue for $\Lambda_1+\lambda^{-\frac{1}{h}}M(\lambda;w_+)$. From \cite[Theorem $6.3.12$]{horn12} we obtain that for any simple eigenvalue $p$ of $\Lambda_1$, then there exists a unique eigenvalue $p(\lambda)$ of $\Lambda_1+\lambda^{-\frac{1}{h}} M(\lambda;w_+)$, admitting the expansion
$$p(\lambda)=p+O(\lambda^{-\frac{2}{h}})\qquad \text{ as } \lambda\to \infty.$$
Therefore $\lambda^{\frac{1}{h}}p(\lambda)$ is an eigenvalue of $\Lambda_\lambda+w_+$, with the asymptotic behaviour
$$\lambda^{\frac{1}{h}}p(\lambda)=\lambda^{\frac{1}{h}}p+O(\lambda^{-\frac{1}{h}})\qquad \text{ as } \lambda\to \infty.$$
Since  $\Lambda_\lambda+w_+\subset I+\mf{b}$, then $\Lambda_\lambda+w_+$ is a regular element \cite{Ko78}. 
Then, its characteristic polynomial in the adjoint representation is given by
\beq\label{charLaQ}
\det(p\, \mbox{Id}+\Lambda_\lambda+w_+)=
p^n\sum_{i=1}^{nh}(p-\lambda^{\frac{1}{h}}p_i(\lambda)),
\eeq
where $p_i(\lambda)=p_i+O(\lambda^{-\frac{2}{h}})$ and the $p_i$ are simple eigenvalues of $\Lambda_1$. 
From the representation \eqref{charLaQ}, it follows that it is a polynomial in $\lambda$ of degree $n$.

\subsection*{Acknowledgments} 
A.R. and M.P. thank  the {Dipartimento di Matematica e Applicazioni\/}  of Universit\`a Milano-Bicocca  for its hospitality.
This work was carried out under the auspices of the GNFM Section of INdAM. We would like to thank the anonymous referee 
whose attentive reading of the manuscript helped to improve the presentation.

\end{document}